\documentclass[reqno,11pt]{amsart}
\usepackage[linktocpage]{hyperref}
\usepackage[utf8]{inputenc}
\usepackage{graphicx}
\usepackage{slashed}
\usepackage{amscd}
\usepackage{amssymb}
\usepackage{esint}
\usepackage{enumitem}
\usepackage{comment}
\usepackage{amsmath}
\usepackage{dsfont}
\usepackage[mathscr]{eucal}
\usepackage{epsfig}
\usepackage{pstricks}
\usepackage{pst-all}

\numberwithin{equation}{section}
\allowdisplaybreaks[1]
\definecolor{labelkey}{gray}{.65}

\title[The Entanglement Entropy of Causal Diamonds]
{The Fermionic Entanglement Entropy of Causal Diamonds in Two-Dimensional Minkowski Space}

\author[F.\ Finster]{Felix Finster}
\address{Fakult\"at f\"ur Mathematik  Universit\"at Regensburg  D-93040 Regensburg  Germany}
\email{finster@ur.de}

\author[M.\ Lottner]{Magdalena Lottner}
\address{Fakult\"at f\"ur Mathematik  Universit\"at Regensburg  D-93040 Regensburg  Germany}
\email{magdalena.lottner@freenet.de}

\author[A.\ Much]{Albert Much}
\address{Institut f\"ur Theoretische Physik\\ Universit\"at Leipzig\\ D-04103 Leipzig \\ Germany}
\email{much@itp.uni-leipzig.de}

\author[S.\ Murro]{Simone Murro \\ \\ July 2024  / September 2025}
\address{Dipartimento di Matematica, Universit\`a di Genova, I-16146 Genova Italy}
\email{simone.murro@unige.it}

\newtheorem{Def}{Definition}[section]
\newtheorem{Thm}[Def]{Theorem}
\newtheorem{Prp}[Def]{Proposition}
\newtheorem{Lemma}[Def]{Lemma}

\newtheorem{Example}[Def]{Example}

\newtheorem{Condition}[Def]{Condition}
\newcommand{\Thanks}{\vspace*{.5em} \noindent \thanks}
\newcommand{\beq}{\begin{equation}}
	\newcommand{\eeq}{\end{equation}}
\newcommand{\Proof}{\begin{proof}}
	\newcommand{\QED}{\end{proof} \noindent}

\newcommand{\la}{\langle}
\newcommand{\ra}{\rangle}

\newcommand{\Sl}{\mbox{$\prec \!\!$ \nolinebreak}}
\newcommand{\Sr}{\mbox{\nolinebreak $\succ$}}

\newcommand{\C}{\mathbb{C}}
\newcommand{\R}{\mathbb{R}}
\newcommand{\1}{\mathds{1}}
\newcommand{\Z}{\mathbb{Z}}

\DeclareMathOperator{\tr}{tr}

\newcommand{\T}{{\mathcal{T}}}

\newcommand{\Cisc}{C^\infty_{\text{\rm{sc}}}}

\newcommand{\Dir}{{\mathcal{D}}}
\renewcommand{\H}{\mathscr{H}}

\DeclareMathOperator{\supp}{supp}

\newcommand{\scrM}{\mycal M}

\newcommand{\scrD}{\mycal D}
\newcommand{\scrN}{\mycal N}

\newcommand{\bitem}{\begin{itemize}[leftmargin=2.5em]}
	\newcommand{\eitem}{\end{itemize}}

\newcommand{\loc}{\text{\rm{loc}}}

\newcommand{\Opa}{\mathrm{Op}_\alpha}
\newcommand{\Op}{\mathrm{Op}}

\newcommand {\BS}{\mathbf S}
\newcommand{\bl}{{\,\vrule depth3pt height9pt}{\vrule depth3pt height9pt}
	{\vrule depth3pt height9pt}{\vrule depth3pt height9pt}\,}

\newcommand {\bxi}{\xi}
\newcommand {\boldeta}{\eta}

\newcommand{\Renyi}{R{\'e}nyi }
\newcommand{\CA}{\mathcal A}
\newcommand{\CB}{\mathcal B}
\DeclareFontFamily{OT1}{rsfso}{}
\DeclareFontShape{OT1}{rsfso}{m}{n}{ <-7> rsfso5 <7-10> rsfso7 <10-> rsfso10}{}
\DeclareMathAlphabet{\mycal}{OT1}{rsfso}{m}{n}

\newcommand{\PiRI}{\Pi^{(\varepsilon)}}

\newcommand{\SN}{\BS}

\begin{document}
\maketitle

\begin{abstract}
The fermionic R\'enyi entanglement entropy is studied for causal diamonds in two-dimensional Minkowski space. Choosing the quasi-free state describing the Minkowski vacuum with an ultraviolet regularization, a logarithmically enhanced area law is derived.
\end{abstract}

\tableofcontents

\section{Introduction} \label{secintro}
Entropy quantifies the disorder of a physical system.
There are various notions of entropy, like the entropy in classical statistical mechanics
as introduced by Boltzmann and Gibbs, the Shannon and R{\'e}nyi entropies in information theory
or the {\em{von Neumann entropy}} for quantum systems.
In the past years, many studies have been devoted to the {\em{entanglement entropy}},
being a measure for the quantum correlations between subsystems of a
composite quantum system~\cite{amico-fazio, horodecki}.
In the relativistic setting, the {\em{relative entropy}} has been studied extensively
in connection with modular theory
(see for example~\cite{hollands-sanders, galanda2023relative, witten-entangle}).
In the present paper we restrict attention to the {\em{fermionic}} case.
Moreover, for simplicity we only consider the {\em{quasi-free}} case where the particles do not
interact with each other. This makes it possible to express the entanglement entropy
in terms of the reduced one-particle density operator~\cite{helling-leschke-spitzer}
(for details in an expository style see the survey paper~\cite{fermientropy}).
Based on methods first developed in~\cite{widom1},
this setting has been studied extensively for a free Fermi gas formed of non-relativistic spinless
particles~\cite{helling-leschke-spitzer, leschke-sobolev-spitzer, LSS_2022}.
The main interest of these studies lies in the derivation of {\em{area laws}}, which quantify
how the entanglement entropy scales as a function of the size of the spatial region forming the
subsystem. More recently, these methods and results were adapted to the
relativistic setting of the Dirac equation. In~\cite{arealaw} the entanglement entropy for the free Dirac field in
a bounded spatial region of Minkowski space is studied.
An area law is proven in two limiting cases: that the volume tends to infinity
and that the regularization is removed. Moreover, in~\cite{bhentropy}
the geometry of a Schwarzschild black hole is studied. The entanglement entropy of the event horizon is
computed to be a prefactor times the number of occupied angular momentum modes.
Independently, the entanglement entropy for systems of Dirac spinors has been
studied in~\cite{bollmann-mueller, bollmann-mueller2}.

In the present paper we study the entanglement entropy of a {\em{causal diamond}}~$\scrD$
embedded in {\em{two-dimensional Minkowski space}}~$\scrM$
(see Figure~\ref{figpos1}).
\begin{figure}
\psset{xunit=.5pt,yunit=.5pt,runit=.5pt}
\begin{pspicture}(251.00544607,225.44960322)
{
\newrgbcolor{curcolor}{0.80000001 0.80000001 0.80000001}
\pscustom[linestyle=none,fillstyle=solid,fillcolor=curcolor]
{
\newpath
\moveto(18.27661984,112.62883094)
\lineto(125.50273134,221.82610968)
\lineto(225.86058709,114.56064952)
\lineto(125.50273134,6.05324039)
\closepath
}
}
{
\newrgbcolor{curcolor}{0 0 0}
\pscustom[linewidth=0.99999871,linecolor=curcolor]
{
\newpath
\moveto(18.27661984,112.62883094)
\lineto(125.50273134,221.82610968)
\lineto(225.86058709,114.56064952)
\lineto(125.50273134,6.05324039)
\closepath
}
}
{
\newrgbcolor{curcolor}{0 0 0}
\pscustom[linewidth=2.4999988,linecolor=curcolor]
{
\newpath
\moveto(0.01077543,112.56283661)
\lineto(235.25604661,114.59081771)
}
}
{
\newrgbcolor{curcolor}{0 0 0}
\pscustom[linestyle=none,fillstyle=solid,fillcolor=curcolor]
{
\newpath
\moveto(231.69583565,121.56038291)
\lineto(251.00545383,114.72658877)
\lineto(231.81652104,107.56090982)
\curveto(235.28123969,111.66842771)(235.2310044,117.49570838)(231.69583565,121.56038291)
\closepath
}
}
{
\newrgbcolor{curcolor}{0 0 0}
\pscustom[linewidth=2.4999988,linecolor=curcolor]
{
\newpath
\moveto(19.27655055,0.01015)
\lineto(17.57166236,209.70013047)
}
}
{
\newrgbcolor{curcolor}{0 0 0}
\pscustom[linestyle=none,fillstyle=solid,fillcolor=curcolor]
{
\newpath
\moveto(10.60035294,206.14333611)
\lineto(17.443611,225.44960234)
\lineto(24.5998835,206.25715955)
\curveto(20.49406437,209.72389111)(14.66675977,209.6765121)(10.60035294,206.14333611)
\closepath
}
}
{
\newrgbcolor{curcolor}{0 0 0}
\pscustom[linewidth=1.75748033,linecolor=curcolor]
{
\newpath
\moveto(225.48306142,128.94234149)
\lineto(225.66474331,98.23634464)
}
\rput[bl](28,210){$t$}
\rput[bl](240,125){$x$}
\rput[bl](219,80){$\lambda$}
\rput[bl](118,91){$\Lambda$}
}
\end{pspicture}
\caption{A causal diamond.}
\label{figpos1}
\end{figure}
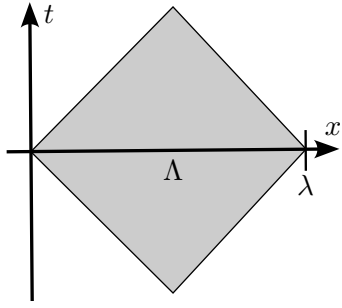%
The interval~$\Lambda := (0, \lambda)$ with~$\lambda>0$ is our spatial subregion.
Its boundary consists of the two corners~$(0,0)$ and~$(0, \lambda)$ of the causal diamond.
In this setting, an area law simply states that the entanglement entropy should be independent
of the size~$\lambda$ of the spatial subregion.
In order to make this statement precise, as the fermionic state in Minkowski space we choose the
vacuum state with an ultraviolet regularization on the scale~$\varepsilon>0$.
More precisely, on the Hilbert space~$\H_\scrM := L^2(\R, \C^2)$ of Dirac wave functions
at time zero,
we consider the bounded pseudo-differential operator~$\PiRI$ defined by
\[ 
(\PiRI \psi)(x) =  \frac{1}{2\pi} \int_{-\infty}^\infty \int_{-\infty}^\infty e^{ik (x-y)}
e^{-\varepsilon \omega(k)}
\bigg(1- \frac{1}{ \omega(k)} \begin{pmatrix} -k & m \\ m & k \end{pmatrix} \bigg)\:
\psi(y) \: dy\,dk \,,
\]
where~$\omega(k):=  \sqrt{k^2+m^2}$ with a mass parameter~$m \geq 0$ and~$\varepsilon>0$.
The parameter~$\varepsilon$ can be interpreted as a semi-classical parameter
that will tend to zero in our asymptotic results.
Now, for each~$\varkappa >0$ we introduce the \textit{R\'enyi entropy function}
as follows. If~$t\notin (0,1)$, we set~$\eta_\varkappa(t) = 0$. For~$t\in (0, 1)$ we define
\begin{align}
\begin{split}
	\label{eq:eta_gamma}
	\eta_\varkappa(t)= & \ 
			\displaystyle \frac{1}{1-\varkappa}\ln \big( t^\varkappa + (1-t)^\varkappa \big)
			\qquad\;\;\, \text{for } \varkappa\neq 1\:,\\[0.2cm]
\eta_1(t):=\lim\limits_{\varkappa \rightarrow 1} \eta_\varkappa(t)
	= &\ -t\ln t - (1-t)\ln (1-t) \qquad \text{for }  \varkappa = 1\;.
\end{split}
\end{align}
Note that~$\eta_1$ is the familiar von Neumann entropy function. 
The {\em{\Renyi entanglement entropy}} of the causal diamond is defined by
\[ 
S_\varkappa(\PiRI, \Lambda; \eta_\varkappa):=\tr \Big( \eta_\varkappa \big( \chi_\Lambda \PiRI \chi_\Lambda \big)
- \chi_\Lambda \,\eta_\varkappa \big( \PiRI \big)\, \chi_\Lambda \Big) \:, \]
where~$\chi_\Lambda$ is the characteristic function of the interval~$\Lambda$.
As we shall see, since~$\Lambda$ is bounded and~$\varepsilon > 0$, both operators on the right-hand side are
trace class, so that the entropy~$S_\varkappa$ is well-defined. Our main objective is to analyze the asymptotic
behavior of the entropy~$S_\varkappa(\PiRI, \Lambda; \eta_\varkappa)$ in the limit~$\varepsilon \searrow 0$.
This is our main result:
\begin{Thm} {\bf{(Entanglement entropy of a causal diamond)}} \label{thmmain}
	\label{thm:MainRindlerSpacetime}
	For any~$\varkappa >0$, the \Renyi entanglement entropy of a causal diamond in two-dimensional Minkowski space satisfies the relation
	\begin{align}
		\label{eq:RenyiEntHalfspace}
		\lim\limits_{\varepsilon \searrow 0} \:\frac{1}{\ln(1/\varepsilon) } S_\varkappa(\PiRI, \Lambda; \eta_\varkappa) =  \frac{1}{ \pi^2} \int_0^1 \frac{\eta_\varkappa(t)}{t(1-t)} \:dt = \frac{1}{6}\frac{\varkappa+1}{\varkappa}\:.
	\end{align}
		In particular, for~$\varkappa =1$, it holds
		\[
		\lim\limits_{\varepsilon \searrow 0} \:\frac{1}{ \ln(1/\varepsilon)} S_\varkappa(\PiRI, \Lambda; \eta) = \frac{1}{3}\:.\]
\end{Thm} \noindent
The fact that the entropy does not depend on the size~$\lambda$ of the diamond
and grows logarithmically as~$\varepsilon$ tends to zero
can be understood as a {\em{logarithmically enhanced area law}}.

Our method of proof is based on extensions of 
the methods in~\cite{widom1, helling-leschke-spitzer, leschke-sobolev-spitzer, LSS_2022}
to matrix-valued symbols as developed in~\cite{bhentropy, arealaw}.
Our presentation is self-contained and of expository style.

The paper is organized as follows. Section~\ref{secprelim} provides the physical preliminaries on the entanglement entropy of the Minkowski vacuum state restricted to a causal diamond.
In Section~\ref{sec:schatten},  we collect some mathematical background of our analysis including some abstract results on on Schatten norms and compact operators on Hilbert spaces.
Our main mathematical results are contained in Section~\ref{sec:truncated}, devoted to an asymp\-to\-tic analysis of truncated pseudo-differential operators.
The main result of this paper is then proved in Section~\ref{sec:diamond}.

\section{Physical Preliminaries} \label{secprelim}

\subsection{The Dirac Field in Two-Dimensional Minkowski Space} \label{secdir2d}
We consider two-dimensional Minkowski space~$\scrM:=(\R^2, g)$ endowed with the line element
\[ 
ds^2 = g_{ij}\: dx^i dx^j = dt^2 - dx^2 \:, \]
and denote by~$S \scrM = \scrM \times \C^2$ the trivial spinor bundle. As customary, we
equip the spinor bundle with the {\em spin inner product}, namely the indefinite inner product
\[ 
\Sl \psi | \phi \Sr = \la \psi, \begin{pmatrix} 0 & 1 \\ 1 & 0 \end{pmatrix} \phi \ra_{\C^2} \:, \]
where~$\la .,. \ra_{\C^2}$ is the canonical scalar product on~$\C^2$.
The {\em{Dirac operator}} is the first order differential operator acting on sections of the spinor bundle defined by
\beq \label{Dirop}
\Dir := i \gamma^j \partial_j \:,
\eeq
where the Dirac matrices~$\gamma^j$ are given, in the chiral representation, by
\[ 
\gamma^0 = \begin{pmatrix} 0 & 1 \\ 1 & 0 \end{pmatrix} \:,\qquad
\gamma^1 = \begin{pmatrix} 0 & 1 \\ -1 & 0 \end{pmatrix} \:. \]
Fixing a mass~$m \geq 0$, the {\em{Dirac equation}} reads
\beq \label{Direq}
(\Dir-m) \psi = 0 \:.
\eeq
The Dirac equation can be rewritten as a symmetric hyperbolic system, showing that its Cauchy problem is
well-posed
(for details see for example~\cite[Chapter~13]{intro}). Moreover, choosing
smooth and compactly supported initial data on a Cauchy surface~$\scrN$, a Dirac solution lies in the class~$\Cisc(\scrM, S\scrM)$
of smooth spinors with spatially compact support. On solutions~$\psi, \phi$
in this class, one introduces the (positive definite) scalar product
\[ 
(\psi | \phi)_\scrM := \int_\scrN \Sl \psi | \slashed{\nu} \phi\Sr|_q\: d\mu_\scrN(q) \:, \]
where~$\slashed{\nu} = \gamma^j \nu_j$ denotes Clifford multiplication by the future-directed unit
normal~$\nu$, and~$d\mu_\scrN$ is the volume measure of the induced Riemannian metric
on~$\scrN$ (thus for the above ray~$\scrN = \{ (\alpha x, x) \text{ with } x > 0 \}$,
the measure~$d\mu_\scrN= \sqrt{1-\alpha^2}\, dx$ is a multiple of the Lebesgue measure).
Due to current conservation, this scalar product is independent of the choice of~$\scrN$.
Forming the completion, we obtain the Hilbert space~$(\H_\scrM, (.|.)_\scrM)$, referred to as the
{\em{solution space}} of the Dirac equation.
For convenience, we always
choose~$\scrN$ as the Cauchy surface~$\{t=0\}$, so that
\[ 
(\psi | \phi)_\scrM = \int_{-\infty}^\infty \Sl \psi | \gamma^0 \phi \Sr|_{(0,x)}\: dx \:. \]

\subsection{The Quantized Dirac Field and its Vacuum State}
The quantized Dirac field can be described in a two step-procedure. First,
one assigns to a classical physical system (described by the solution space) a unital $*$-algebra~$\mathfrak A$, whose elements are interpreted as observables of the system at hand. Then, one
determines the admissible physical states of the system by identifying a suitable subclass of
the linear, positive and normalized functionals~$\omega : \mathfrak A \to \C$. Once a state is specified, the Gelfand-Naimark-Segal (GNS) construction guarantees the existence of a representation of the quantum
field algebra as (in general, unbounded) operators defined in a common dense subspace of some Hilbert space.

We here restrict attention to a {\em{quasi-free}} Dirac state, which is fully characterized by its two-point
distribution. Then, as shown in~\cite[Section 4]{capoferri-murro}, 
the above procedure boils down to constructing suitable self-adjoint operators on the Dirac solution space.
In particular, the {\em{vacuum state}} can be represented by the projection operator
onto the space of negative frequencies solutions as follows. We rewrite the Dirac equation in the Hamiltonian form
\[ i \partial_t \psi = H \psi \qquad \text{with} \qquad
H = -i \gamma^0 \gamma^1 \partial_x + m \gamma^0 = \begin{pmatrix} i \partial_x & m \\ m & -i \partial_x \end{pmatrix} \:. \]
Taking the spatial Fourier transform
\[ \hat{\psi}(k) = \int_{-\infty}^\infty \psi(x)\: e^{-i k x} \:dx\:, \qquad
\psi(x) = \int_{-\infty}^\infty \frac{dk}{2 \pi}\: \hat{\psi}(k)\: e^{i k x} \:, \]
the Hamiltonian becomes the multiplication operator
\[ \hat{H}(k) = \begin{pmatrix} -k & m \\ m & k \end{pmatrix} \:, \]
which can be diagonalized by
\[ \hat{H}(k) = \omega E_+ - \omega E_- \:, \]
with~$\omega(k) := \sqrt{k^2+m^2}$ and
\[ E_\pm(k) =  \frac{1}{2} \pm \frac{1}{2 \omega(k)} \begin{pmatrix} -k & m \\ m & k \end{pmatrix} \:. \]
The projectors onto the space of negative frequencies solutions is 
\[ (\Pi \psi)(x) = \int_{-\infty}^\infty \Pi(x,y)\: \psi(y)\: dy \]
with integral kernel
\[ \Pi(x,y) =  \frac{1}{2} \int_{-\infty}^\infty \frac{dk}{2 \pi}\:
\bigg(\1_{\C^2} - \frac{1}{  \omega(k)} \begin{pmatrix} -k & m \\ m & k \end{pmatrix} \bigg)\:
e^{ik (x-y)} \:. \]

For future convenience, we define the {\em{regularized projection operator}} onto the negative
frequencies solutions as the integral operator
\[ (\PiRI \psi)(x) = \int_{-\infty}^\infty \PiRI(x,y)\: \psi(y)\: dy \:, \]
where the integral kernel is given by
\[ \PiRI(x,y) =  \frac{1}{2} \int_{-\infty}^\infty \frac{dk}{2 \pi}\:
e^{-\varepsilon \omega(k)}
\bigg(\1_{\C^2} - \frac{1}{  \omega(k)} \begin{pmatrix} -k & m \\ m & k \end{pmatrix} \bigg)\:
e^{ik (x-y)} \:. \]
This operator will play a pivotal role for computing the entanglement entropy.

\subsection{The Entanglement Entropy of a Causal Diamond}\label{secentquasi}
In this section we shall define the \Renyi entanglement entropy of a causal diamond~$\scrD$
embedded in two-dimensional Minkowski space~$\scrM$. For the sake of completeness, let us recall that a causal diamond is a two-dimensional spacetime isometric to the subset of two-dimensional Minkowski space
\[ 
\scrD = \big\{ (t,x) \in \scrM \;\;\;\text{with}\;\;\; 0 < x < \lambda \text{ and } |t| < \min(x, \lambda-x) \big\}\:; \]
where~$\lambda>0$ is an arbitrary, but fixed parameter (see Figure~\ref{figpos1}).
Then the inclusions
\[ \scrD \subset \scrM \qquad \text{and} \qquad S\scrD = \scrD \times \C^2
\subset \scrM \times \C^2 = S\scrM \]
are clearly isometries, and the Dirac operator and the Dirac equation are again given by~\eqref{Dirop}
and~\eqref{Direq}. Adopting the notation of the previous section, we denote the subspace of solutions in~$\scrD$ by~$\H_\scrD$. As shown in\cite[Appendix A]{fermientropy}, the 
\Renyi entanglement entropy of~$\scrD$ can be expressed as
\[ S_\varkappa(\PiRI, \Lambda,\eta_\varkappa)=  \tr \eta_\varkappa\big( \pi_\scrD \PiRI \pi_\scrD \big)
- \tr \eta_\varkappa(\PiRI) \:, \]
where~$\pi_\scrD : \H \rightarrow \H_\scrD$ is the orthogonal projection operator.
This projection operator can be represented more
concretely as the multiplication operator by a characteristic function acting on the wave functions
on the Cauchy surface~$\{t=0\}$, i.e.\
\[ (\pi_\scrD \psi)(0,x) = \chi_\Lambda(x)\: \psi(0,x) \:, \]
where~$\Lambda = (0, \lambda)$ for the causal diamond (see again Figure~\ref{figpos1}).

\section{Schatten-von-Neumann Bounds for Pseudo-Differential Operators}
\label{sec:schatten}
In this section we state some basic
definitions and results on singular values and Schatten-von Neumann classes. 
For more details we refer to \cite[Chapter~11]{birman-solomjak}.

\subsection{Singular Values and Schatten-von Neumann Classes} 
\label{sec:SingValuesSchatten}
For a compact operator~$A$ in a separable Hilbert space~$\H$ we denote 
by~$s_k(A),\: k = 1, 2, \dots,$ its singular values (defined as the eigenvalues 
of the self-adjoint compact operator~$|A|$) labelled in non-increasing order counting 
multiplicities. For the sum~$A+B$ the following inequality holds,
\begin{align}
	\label{ineq:s_k A+B}
	s_{2k}(A+B)\le s_{2k-1}(A+B)\le s_k(A) + s_k(B) \:. 
\end{align}
We say that~$A$ belongs to the Schatten-von Neumann class~$\BS_p$, $p>0$, 
if the series
\[ \|A\|_{p} := \big(\tr|A|^p\big)^{1/p} = \bigg(\sum_{n=1}^\infty s_k(A)^p \bigg)^{1/p} \]
is finite. The functional 
$\|A\|_{p}$ defines a norm if~$p\ge 1$ 
and a quasi-norm if~$0<p<1$. With this (quasi-)norm, the class~$\BS_p$ is a complete space. 
Note that for~$p=1$ this coincides with the trace norm. Moreover, by~$\|.\|_{\infty}$ we denote the ordinary operator norm.
For~$0<p<1$ the quasi-norm is actually a \emph{$p$-norm}, that is,  
it satisfies the following ``triangle inequality" for all~$A, B\in\BS_p$:
\begin{align}\label{eq:ptriangl}
	\|A+B\|_{p}^p\le \|A\|_{p}^p + \|B\|_{p}^p\,.
\end{align} 
This inequality will be used frequently in what follows.
We also point out a useful estimate 
for individual eigenvalues for operators in~$\BS_p$:
\begin{align}\label{eq:ind}
	s_k(A)\le k^{-\frac{1}{p}}\, \|A\|_p,\ k = 1, 2, \dots.
\end{align}
In~\cite[p.~262]{birman-solomjak} 
it is shown that the norms~$\|\cdot \|_{\BS_p}$ also fulfill a H\"older-like inequality, meaning that for any~$0<p<\infty$ and~$0<p_1,p_2\leq \infty$ such that~$p^{-1}=p_1^{-1}+p_2^{-1}$ and~$A \in \BS_{p_1}$, $B \in \BS_{p_2}$, the operator~$AB \in \BS_p$ with
\begin{align}
	\label{ineq:HoelderLike}
	\|AB \|_{p}  \leq  \|A \|_{p_1} \|B \|_{p_2} \:,
\end{align}
where by~$\|.\|_\infty$ we mean the ordinary operator norm.
Moreover, as explained in~\cite[p.~254]{birman-solomjak}, for any two~$0< p_1 <p_2 \leq \infty$, we
have~$\BS_{p_1}\subset \BS_{p_2}$ and for any~$A\in \BS_{p_1}$
\[ 
	\| A \|_{p_2} \leq \| A \|_{p_1}\:. \]

We refer to~\cite[Chapter~11]{birman-solomjak} for more details on singular values.

\subsection{Non-Smooth Spectral Functions}
We now state a result for compact operators on an arbitrary separable Hilbert space~$\H$. 
Let~$A$ be a symmetric bounded operator on~$\H$ and~$P$ an orthogonal projection operator on~$\H$. Given a
continuous function~$f \in C(\R)$ we define the operator 
\begin{align*}
	D(A, P; f) := P f(PAP) P - P f(A) P.
\end{align*}
In what follows it is convenient to require that the function~$f$ satisfies the following condition.

 \begin{Condition}\label{cond:f}
The function~$f\in C^2(\R\setminus\{ t_0 \})\cap C{}(\R)$ satisfies the bound 
\[ 
	\bl f\bl_2 := \max_{0\le k\le 2}\sup_{t\not = t_0} |f^{(k)}(t)| |t-t_0|^{-\gamma+k}<\infty \]
for some~$\gamma\in (0, 1]$ and is supported on the interval~$(t_0-R, t_0+R)$ with some finite 
$R>0$.  
\end{Condition}

\begin{Example}
\label{ex:EtaCondf}
Consider the functions~$\eta_\varkappa$ defined in~\eqref{eq:eta_gamma} and set~$\T=\{0,1\}$. Then,  in the neighborhood of every~$t_i\in\T$, 
there exist positive constants~$\gamma$ and~$c_k>0$ with~$k=0,1,2$ such that
\[ |\eta_\varkappa^{(k)}(t)| \leq c_k \:|t-t_i|^{\gamma -k}\,. \]
As shown in~\cite[Lemma~D.1]{bhentropy}, the value of~$\gamma$ depends on~$\varkappa$ as follows,
\[ \begin{cases}
\gamma\leq \min\{1,\varkappa\} & \text{for }\varkappa\neq 1\\
\gamma <1 & \text{for }\varkappa = 1 \:.
\end{cases} \]
Notice that, using a partition of unity~$(\psi_k)_{0\leq k\leq 1}$ such that the support of each~$\psi_k$ only contains exactly one the elements in~$\T$, then each~$\eta_\varkappa \psi_k$ satisfies Condition~\ref{cond:f}.
\end{Example}

The next proposition follows from a more general result proven in 
\cite[Theorem~2.4]{sobolev-functions}; see also~\cite[Proposition 2.2]{leschke-sobolev-spitzer}.

\begin{Prp}\label{prop:szego} 
Suppose that~$f$ satisfies Condition~\ref{cond:f}  
for some~$\gamma\in (0, 1]$ and 
some~$t_0\in\R$ with~$R>0$. Let~$q \in (0, 1]$ and assume that~$\sigma< \min(2-q^{-1}, \gamma)$.  
Let~$A, B$ be two bounded self-adjoint operators and assume that~$|A - B|^\sigma \in \BS_p$. Then
\[ \|f(A) - f(B)\|_p \lesssim  \bl f\bl_n \,R^{\gamma-\sigma} \,\big\| |A - B|^\sigma \big\|_p \]
with a positive implicit constant independent of~$A$, $B$, $f$ and $R$. In particular, for an orthogonal projection~$P$ 	such that~$PA(I-P)\in \BS_{\sigma q}$ it holds
\begin{equation}\label{eq:szego}
	\|  D(A, P; f) \|_{q}
	\lesssim \bl f\bl_2\, R^{\gamma - \sigma} \, \|P A ( I-P)\|_{{\sigma q}}^\sigma \:,
\end{equation}
with a positive implicit constant 
independent of the operators~$A, P$, the function~$f$, and the parameter~$R$.  
\end{Prp}

\section{Spectral Analysis of Truncated Pseudo-Differential Operators}\label{sec:truncated}
Using the results from the previous section, we are now in the position to perform an asymptotic analysis of truncated pseudo-differential operators. To this end, consider a pseudo-differential operator defined as usual by
\[ 
\big(\Opa(\CA) \psi \big) (x)  = \frac{\alpha}{2\pi} \int_{-\infty}^\infty  \int_{-\infty}^\infty  \: e^{i\alpha \xi(x-y)} \CA(\xi) \psi(y)\ dy \,d\xi\:, \]
where~$\mathcal{A}$ is a $2\times 2$ matrix-valued symbol and~$\alpha$ is a strictly positive constant.  
A {\em truncated} pseudo-differential operator is then obtained multiplying~$\Opa(\CA)$ by the characteristic function~$\chi_\Lambda$ of~$\Lambda\subset \R$,
\[ \chi_\Lambda \Opa(\CA)  \chi_\Lambda \:. \]
We define the {\em entropic difference operator} by
\[ D_\alpha ( \CA, \Lambda; f ) := f \big( \chi_\Lambda \Opa(\CA)  \chi_\Lambda \big) - \chi_\Lambda f(\Opa(\CA) ) \chi_\Lambda \,. \]
Throughout this section we assume that the function~$f$ satisfies Condition~\ref{cond:f}. Inspired by~\cite[Lemma~5.6]{arealaw}, this is our first result.

\begin{Lemma}
	\label{lem:ErrorTermsAbstr}
Suppose that~$f$ satisfies Condition~\ref{cond:f} for some~$\gamma \in (0, 1]$. Let	$q \in (0, 1)$
and assume that   
	$\sigma< \min(2-q^{-1}, \gamma)$. Finally let be~$\CA^{(1)}_\alpha$ and~$\CA^{(2)}_\alpha$ two families of symbols  satisfying the conditions
\begin{align}
	\label{eq:symbolcon}
	& \sup_{\xi\in \R} \big| \CA^{(1)}_\alpha (\xi) - \CA^{(2)}_\alpha(\xi)\big| \rightarrow 0&&\hspace*{-1cm} \text{as~$\alpha \rightarrow \infty$} \\ 
		\label{cond:chi_Lambda Opa (1-chi_Lambda)}
	&	\Big\| \chi_\Lambda\, \Opa \big( \CA^{(j)}_\alpha \big)\, (1-\chi_\Lambda) \Big\|_{{\sigma q}}^{\sigma q}  \lesssim g(\alpha) &&\hspace*{-1cm}\text{for~$j=1,2$}\:, 
	\end{align}
for some~$q < \gamma$ and some positive function~$g$. Then
	\begin{align*}
		\lim\limits_{\alpha \rightarrow \infty } \frac{1}{g(\alpha)}\, \big\|D_\alpha(\CA^{(1)}_\alpha, \Lambda;f) - D_\alpha(\CA^{(2)}_\alpha, \Lambda;f) \big\|_{1} =0\:.
	\end{align*}
\end{Lemma}
\begin{proof}

For ease of notation, throughout the proof we denote
\[ D_\alpha(\CA^{(j)}_\alpha)  \equiv D_\alpha(\CA^{(j)}_\alpha, \Lambda;f)\:, \qquad j=1,2\:.  \]
Take~$0<\delta <1$ arbitrary.
For any~$\alpha>1$ we define
\[N\equiv N(\alpha ):= \big\lceil g(\alpha)\,\delta^{\frac{q}{q-1}} \big\rceil \:. \]
Now rewrite
\[ \big\| D_\alpha(\CA^{(1)}_\alpha) - D_\alpha(\CA^{(2)}_\alpha) \big\|_{1}
= \sum_{k=1}^{\infty} s_k \big( D_\alpha(\CA^{(1)}_\alpha) - D_\alpha(\CA^{(2)}_\alpha) \big)
= Z_1(N) + Z_2(N)\:, \]
where~$Z_1$ and~$Z_2$ involve the small respectively large singular values,
\begin{align*}
Z_1(N)&:=\sum_{k=1}^{2N} s_k\big( D_\alpha(\CA^{(1)}_\alpha) - D_\alpha(\CA^{(2)}_\alpha)  \big)  \\
Z_2(N)&:= \sum_{k=2N+1}^{\infty} s_k\big( D_\alpha(\CA^{(1)}_\alpha) - D_\alpha(\CA^{(2)}_\alpha)  \big)\:.
\end{align*}

We now estimate~$Z_1$ and~$Z_2$ separately. For the estimate of~$Z_1$, we use that
\[ s_k\big( D_\alpha(\CA^{(1)}_\alpha) - D_\alpha(\CA^{(2)}_\alpha)  \big) \leq \| D_\alpha(\CA^{(1)}_\alpha) - D_\alpha(\CA^{(2)}_\alpha) \|_\infty \:.\]
As shown in
\cite[Lemma 5.5]{arealaw}, since~$\CA^{(1)}_\alpha$ and~$\CA^{(2)}_\alpha$ are two families of uniformly 
bounded self-adjoint operators satisfying~$\|\CA^{(1)}_\alpha-\CA^{(2)}_\alpha\|_\infty\to 0$ 
as~$\alpha \to \infty$, then, for any function~$g\in C(\R)$ we have 
\begin{align*}
	\|g(\CA^{(1)}_\alpha) - g(\CA^{(2)}_\alpha)\|_\infty\to 0 \qquad \textup{as~$\alpha\to \infty$}\:.
\end{align*}	
Therefore, combining this observation with Example~\ref{ex:EtaCondf}, it is clear that there exists~$\tilde{\alpha}(\delta)\equiv \tilde{\alpha}>0$ such that for any~$\alpha > \tilde{\alpha}$,
\begin{align*}
\big\| \eta_\varkappa \big( \chi_\Lambda \Opa(\CA^{(1)}_\alpha) \chi_\Lambda \big) - \eta_\varkappa \big( \chi_\Lambda \Opa(\CA^{(2)}_\alpha) \chi_\Lambda \big) \big\|_\infty &\leq \delta^{\frac{1}{1-q}}\qquad \text{and}\\
\big\| \chi_\Lambda \big( \eta_\varkappa\big( \Opa(\CA^{(1)}_\alpha) \big) -  \eta_\varkappa\big( \Opa(\CA^{(2)}_\alpha) \big) \big) \chi_\Lambda \big\|_\infty &\leq \delta^{\frac{1}{1-q}}\:.
\end{align*}
Thus for any~$\alpha > \tilde{\alpha}$ we obtain
\[ Z_1(N) \leq 2N \| D_\alpha(\CA^{(1)}_\alpha) - D_\alpha(\CA^{(2)}_\alpha) \|_\infty \leq 4 N \delta^{\frac{1}{1-q}} \leq 4 g(\alpha)(\delta+\delta^{\frac{1}{1-q}})\leq 8 g(\alpha)\, \delta\:. \]

In order to estimate~$Z_2$, we use~\eqref{ineq:s_k A+B} and~\eqref{eq:ind} to obtain
\begin{align*}
Z_2(N) &\leq 2\sum_{k=N}^\infty s_k\big(D_\alpha(\CA^{(1)}_\alpha)\big) +  2\sum_{k=N}^\infty s_k\big(D_\alpha(\CA^{(2)}_\alpha)\big) \\
&\leq 2 \Big( \big\| D_\alpha(\CA^{(1)}_\alpha) \big\|_{q} + \big\| D_\alpha(\CA^{(2)}_\alpha ) \big\|_{q} \Big) \sum_{k=N}^\infty k^{-1/q} \:.
\end{align*}
We now notice that, applying Proposition~\ref{prop:szego} to~$A= \Opa(\CA)~$ and~$P = \chi_\Lambda$
we  obtain
\[ 
	\|  D_\alpha(\CA, \chi_\Lambda; f) \|_{q}
	\lesssim \bl f\bl_2\, R^{\gamma - \sigma} \, \|\chi_\Lambda \Opa(\CA) (\1- \chi_\Lambda)\|_{\sigma q}^\sigma \]
with a positive implicit constant 
independent of~$\CA, \Lambda$, the function~$f$, and the parameter~$R$.  Combining the latter estimate with  condition~\eqref{cond:chi_Lambda Opa (1-chi_Lambda)}, it follows that
\[ \|D_\alpha(\CA^{(j)}_\alpha)\|_{q}^q \lesssim g(\alpha) \qquad \text{j=1,2} \:. \]
Consequently, $Z_2(N)$ may be estimated by
\begin{align*}
	Z_2(N) 	&\lesssim g(\alpha)^{1/q} \sum_{k=N}^\infty k^{-1/q} \leq 
	g(\alpha)^{1/q}  \int_{N-1}^{\infty}  k^{-1/q} dk \leq g(\alpha)^{1/q}  \int_{g(\alpha)	\delta^{\frac{q}{q-1}}}^{\infty}  k^{-1/q} dk \\
	&\lesssim
	g(\alpha)^{1/q} \,\Big(g(\alpha) \,\delta^{\frac{q}{q-1}}\Big)^{1-1/q} =  g(\alpha)\, \delta\:.
\end{align*}

In summary, for any~$\alpha > \tilde{\alpha}(\delta)$ we obtain the estimate
\[  
\|D_\alpha(\CA^{(1)}_\alpha) - D_\alpha(\CA^{(2)}_\alpha) \|_{1} \lesssim g(\alpha)\, \delta \:,
\]
which leads to
\[
	\lim\limits_{\alpha \rightarrow \infty } \frac{1}{g(\alpha)}\, \big\|D_\alpha(\CA^{(1)}_\alpha, \Lambda;f) - D_\alpha(\CA^{(2)}_\alpha, \Lambda;f) \big\|_{1}  \leq \delta \:.
\]
Since~$\delta \in(0,1)$ is arbitrary, this completes the proof.
\end{proof}

We also recall the following result stated in~\cite[Proposition~2.14]{bhentropy}, based on
earlier results~\cite[Theorem 11.1]{birman-solomjak2}, \cite[Section 5.8]{birman-karadzov} and \cite[Theorem 4.5]{simon2005}. Given~$\sigma \in (0, 2)$ and~$g \in L^2_\text{loc}(\R)$ we set
\beq \label{sigmanorm}
|g|_\sigma := \bigg[ \sum_{z \in \Z} \bigg( \int_z^{z+1} |g(x)|^2 \bigg)^{\sigma/2} \:\bigg]^{1/\sigma}\:.
\eeq
This defines a norm if~$\sigma \geq 1$.

\begin{Prp}	\label{prp72}
Given a matrix-valued symbol~$\CA \in L^2_{\loc}(\R)$ and a function~$h\in L^2_{\loc}(\R)$ with~$\| \CA\|_\sigma, |h|_\sigma<\infty$ for some~$\sigma \in (0,2)$,
it follows that~$h \,\mathrm{Op}_1(\CA) \in \SN_\sigma$ and
	\[ \big\| h\: \mathrm{Op}_1 (\CA) \big\|_\sigma \leq C \:|h|_\sigma \:|\CA|_\sigma \:.  \]
\end{Prp}

We conclude this section by considering symbols~$\CA$ which satisfy the following condition.
\begin{Condition}
\label{cond:MultiScale}
Consider a smooth matrix-valued symbol~$\CA(\xi)$  
for which there exist positive continuous functions  
$v$ and~$\tau$ such that the following estimate holds,
\begin{equation}\label{scales:eq}
	\|\nabla_{\bxi}^n \CA(\bxi) \|\lesssim
	\tau(\bxi)^{-n} \,v(\bxi)\,,\qquad \text{for all~$\bxi \in \R$ and~$n = 0, 1, \dots\:.$}\:.
\end{equation}
We call~$\tau$ the {\bf{scale}} (function)
and~$v$ the {\bf{amplitude}} (function).
The scale~$\tau$ is assumed to be globally Lipschitz 
with Lipschitz constant~$\nu <1$, that is, 
\begin{equation}\label{Lip:eq}
	|\tau(\bxi) - \tau(\boldeta)| \le \nu\: |\bxi-\boldeta|\qquad \text{for all~$\bxi,\boldeta\in\R$}\:.
\end{equation}
Moreover, there shall also exist constants~$c,C>0$ such that the amplitude~$v$ satisfies the bounds
\begin{equation}\label{w:eq}
	c < \frac{v(\bxi)}{v(\boldeta)}< C \quad \text{ for all }\boldeta\in B\bigl(\bxi, \tau(\bxi)\bigr)\,,
\end{equation}
with~$c$ and~$C$ independent of~$\bxi$ and~$\boldeta$.
It is useful to think of~$v$ and~$\tau$ as 
(functional) parameters. They, in turn, may depend on other 
parameters (e.g.\ numerical parameters like~$\alpha$).
\end{Condition}
Given functions~$\nu$ and~$\tau$ and numbers~$\sigma>0$, $\lambda \in \R$ we denote
\[
V_{\sigma,\lambda}(v,\tau):= \int_{-\infty}^\infty \frac{v(\xi)^\sigma}{\tau(\xi)^\lambda}\: d\xi\:.
\]
The next result is a special case of~\cite[Lemma 3.4]{leschke-sobolev-spitzer}
(the generalization to matrix-valued symbols follows from the triangle inequality~\eqref{eq:ptriangl}).

\begin{Prp} \label{prop:cross_smooth} 
Let~$\Lambda\subset \R$ be a bounded interval and let 
the functions~$\tau$ and~$v$ be as described above. 
Suppose that the symbol~$\CA$ satisfies the bounds~\eqref{scales:eq}, and that the conditions 	
\[ 
\tau_{\textup{\tiny inf}} := \inf_{\bxi\in\R}\tau(\bxi)>0 \qquad \text{and} \qquad
\alpha\,\tau_{\textup{\tiny inf}}\gtrsim 1 \]
hold. Then for any~$r\in (0, 1]$ we have 
\[ 
	\big\| \chi_\Lambda \,\Opa(\CA)\, (1-\chi_\Lambda) \big\|_{r}^r
	\lesssim V_{r,1}(v,\tau)\:. \]
This bound is uniform in the symbols~$\CA$ satisfying Condition~\ref{cond:MultiScale} with the same 
implicit constants as in~\eqref{scales:eq} (namely, the constants~$\nu$ in~\eqref{Lip:eq} and~$c,C$ in~\eqref{w:eq},
but not necessarily the same functions~$\tau$ and~$v$). 
\end{Prp}

\section{An Area Law for the Causal Diamond} \label{sec:diamond}
This section is the core of the paper. Our goal is to study the asymptotic behavior of the operator
\[ S_\varkappa(\PiRI, \Lambda; \eta_\varkappa):=\tr \Big( \eta_\varkappa \big( \chi_\Lambda \PiRI \chi_\Lambda \big) - \chi_\Lambda \eta_\varkappa(\PiRI) \chi_\Lambda \Big) \qquad \text{ as~$\varepsilon \searrow 0$}\:. \]
By construction, the regularized one-particle density operator is the pseudo-differential operator~$\PiRI=\Op_1(\CA^{(\varepsilon)})$ with symbol given by 
\[ \CA^{(\varepsilon)} := e^{-\varepsilon \omega(k)}
\bigg(1- \frac{1}{ \omega(k)} \begin{pmatrix} -k & m \\ m & k \end{pmatrix} \bigg)\,. \]
In order to connect our analysis with the results obtained in the previous sections, we introduce
the parameter~$\alpha:=\varepsilon^{-1}$ and introduce the rescaled momentum variable~$\xi$ by
\beq \label{xiscale}
\xi := \varepsilon k = \frac{k}{\alpha} \:,
\eeq
obtaining
\beq \label{PiOpa}
\PiRI =\Opa(\CA_\alpha)
\eeq
with symbols
\beq \label{Aadef}
\CA_\alpha(\xi) := \frac{1}{2}\:e^{-\sqrt{\xi^2+(m/\alpha)^2}}\:
\bigg[\1_{\C^2} - \frac{1}{  \sqrt{\xi^2 + (m/\alpha)^2}} \begin{pmatrix} -\xi & m/\alpha \\ m/\alpha & \xi \end{pmatrix} \bigg]\,.
\eeq
We are interested in the asymptotics for large~$\alpha$.
In particular, the \Renyi entanglement entropy~$S_\varkappa$ can be rewritten as the trace of the
entropic difference operator defined by
\[ D_\alpha(\eta_\varkappa,\Lambda,\mathcal A) :=  \eta_\varkappa\big( \chi_\Lambda \Opa(\mathcal A) \chi_\Lambda \big) - \chi_\Lambda \eta_\varkappa\big(\Opa(\mathcal A)\big) \chi_\Lambda \:. \]
Since for~$\alpha \rightarrow \infty$ the symbol converges to
\beq \label{Ainf}
A_{\infty}(\xi) := \lim_{\alpha\to \infty} A_{\alpha}(\xi)
= e^{- \vert \xi \vert } \begin{pmatrix}  \chi_{\mathbb{R}^+}(\xi) & 0\\ 0 &  \chi_{\mathbb{R}^-}(\xi) \end{pmatrix} \:,
\eeq
we split the trace of the entropic difference into the two contributions
\begin{align}
 \tr  D_\alpha\big(\eta_\varkappa, \Lambda, \mathcal A_\alpha \big) &= \tr  D_\alpha\big(\eta_\varkappa, \Lambda, \CA_\infty \big) +  \Big(  \tr D_\alpha(\eta_\varkappa, \Lambda, \CA_\alpha ) -  \tr D_\alpha(\eta_\varkappa, \Lambda, \CA_\infty) \Big) \notag \\
 &= \tr  D_\alpha\big(\eta_\varkappa, \Lambda, \CA_\infty \big) +  \tr\Big(   D_\alpha(\eta_\varkappa, \Lambda, \CA_\alpha ) -   D_\alpha(\eta_\varkappa, \Lambda, \CA_\infty) \Big)\:. \label{contribs}
\end{align}

In the following sections, we shall analyze these two contributions separately.
The first term will give rise to the enhanced area law, whereas the second term will tend to zero.
Before entering the details, we remark that in the massless case~$m=0$, the second term vanishes.
Therefore, the proof in the massless case will be completed already at the end of Section~\ref{secdiag}

\subsection{The Diagonal Terms} \label{secdiag}

In this section, we shall estimate the trace of the entropic difference~$ D_\alpha(\eta_\varkappa, \Lambda, \CA_\infty)$. Since the symbol~$\CA_\infty$ is diagonal, our analysis reduces to that in~\cite{bhentropy}.
\begin{Lemma} \label{lemmadiag}
\[ \lim_{\alpha \rightarrow \infty} \frac{1}{\log \alpha}
\tr  D_\alpha(\eta_\varkappa, \Lambda, \CA_\infty)   = \frac{1}{ \pi^2} \int_0^1 \frac{\eta_\varkappa(t)}{t(1-t)} \:dt = \frac{1}{6}\frac{\varkappa+1}{\varkappa}\:. \]
\end{Lemma}
\Proof Having a diagonal symbol, the trace splits into a sum of the entropic differences for scalar-valued symbols. Indeed, setting
\[ a_{1}(\xi):=e^{-\xi}\, \chi_{\mathbb{R}^+}(\xi) \qquad \text{and} \qquad
a_{2}(\xi):=e^{\xi} \,\chi_{\mathbb{R}^-}(\xi) \:, \]
the trace of the entropic difference is given by
\[ \tr  D_\alpha(\eta_\varkappa, \Lambda, \CA_\infty )= \sum_{i=1}^{2}  \tr D_\alpha(\eta, \Lambda, a_{i}) \:. \]
Moreover, the function~$\eta_\varkappa$ satisfies Condition~\ref{cond:f}, as explained in Example~\ref{ex:EtaCondf}. Therefore, we are in the position to apply~\cite[Corollary 5.10]{bhentropy}, obtaining
\begin{align*}
	\lim_{\alpha\to \infty} \frac{1}{\log (\alpha)} \tr D_\alpha(\eta, \Lambda, a_{i})  = \frac{1}{\pi^2}  \,U \big( 1 ; \eta_\varkappa \big)   = \frac{1}{ 2\pi^2} \int_0^1 \frac{\eta_\varkappa(t)}{t(1-t)} \:dt = \frac{1}{12}\frac{\varkappa+1}{\varkappa}\,.
\end{align*}
Summing over~$i$ concludes the proof.
\QED

\subsection{The Off-Diagonal Terms}
It remains to  show that the off-diagonal contribution in~\eqref{contribs} vanishes asymptotically.
More precisely, our task is to show that
\begin{equation}\label{eq:offdiag}
\lim_{\alpha\to\infty}\frac{1}{\log\alpha} \tr\Big(   D_\alpha(\eta_\varkappa, \Lambda, \CA_\alpha ) -   D_\alpha(\eta_\varkappa, \Lambda, \CA_\infty) \Big)=0\,. 
\end{equation}
The main difficulty is related to the fact that the symbols do {\em{not}} converge uniformly in~$\xi$
as~$\alpha \rightarrow \infty$. This is obvious by taking the following limits of~\eqref{Aadef} and~\eqref{Ainf},
\[ \lim_{\alpha \rightarrow \infty} \CA_\alpha(0) = \frac{1}{2}\:
\bigg(\1_{\C^2} - \begin{pmatrix} 0 & 1 \\ 1 & 0 \end{pmatrix} \bigg)
\neq \begin{pmatrix}  1 & 0\\ 0 &  0 \end{pmatrix}
= \lim_{\xi \searrow 0} A_{\infty}(\xi) \:. \]
This difficulty is caused by the region~$|\xi| \lesssim m/\alpha$ of small frequencies.
With this in mind, we must treat the high and low frequencies separately.
We let~$\Theta \in C^\infty_0(\R)$ be a smooth cutoff function with
\[ \supp \Theta \subset [-2,2] \qquad \text{and} \qquad \Theta|_{[0,1]} \equiv 1 \]
and set
\[ \Theta_\alpha(\xi) := \Theta \Big( |\xi|\: \frac{\alpha}{\sqrt{\log \alpha}}\Big)\:. \]
Note that this function is supported for~$|\xi| \leq 2 \sqrt{\log \alpha}/\alpha$.
We decompose the symbols~$\CA_\alpha$ and $\CA_\infty$ as
\[ \CA_\alpha(\xi) = \CA^{>}_\alpha(\xi) + \CA^{<}_\alpha(\xi) \:, \qquad
\CA_\infty(\xi) = \CA^{>}_\infty(\xi) + \CA^{<}_\infty(\xi) \]
with
\begin{align}
\CA^{>}_\bullet(\xi) &:= \big(1 - \Theta_\alpha(\xi) \big)\: \CA_\bullet(\xi) \label{CAG} \\
\CA^{<}_\bullet(\xi) &:= \Theta_\alpha(\xi)\: \CA_\bullet(\xi) \label{CAK}
\end{align}
(where the bullet stands for~$\alpha$ or~$\infty$; for clarity we point out that the
symbols~$\CA^{>}_\infty$ and~$\CA^<_\infty$ depend on~$\alpha$ via the cutoff function~$\Theta_\alpha$).
We next derive the following estimate for the high-frequency contributions.
\begin{Prp}\label{prop:highfreq}
\beq \label{highfreq}
\lim_{\alpha\to\infty}\frac{1}{\log\alpha}\tr \Big( D_\alpha(\eta_\varkappa, \Lambda, \CA_\alpha^> ) -   D_\alpha(\eta_\varkappa, \Lambda, \CA^>_\infty) \Big)=0 \:.
\eeq
\end{Prp}
\begin{proof}
Our claim follows by applying Lemma~\ref{lem:ErrorTermsAbstr}. 
We must verify that all the hypotheses are satisfied. First of all, after treating
the function~$\eta_\varkappa$ defined in~\eqref{eq:eta_gamma} as explained in Example~\ref{ex:EtaCondf},
the Condition~\ref{cond:f} is satisfied with~$\gamma<\min\{1,\varkappa\}$. Next, since the function
\[ \big( \CA^{>}_\alpha - \CA^>_\infty \big)(\xi) = \big(1 - \Theta_\alpha(\xi) \big)\: \big( \CA_\alpha - \CA^>_\infty \big)(\xi) \]
is supported outside the problematic region~$|\xi| \leq 2 \sqrt{\log \alpha}/\alpha$, we have uniform convergence,
\[ \CA^{>}_\alpha - \CA^>_\infty \quad \text{converges uniformly to zero as~$\alpha \rightarrow \infty$} \:, \]	
so that Condition~\eqref{eq:symbolcon} is fulfilled.

It remains to show that Condition~\eqref{cond:chi_Lambda Opa (1-chi_Lambda)} holds.
To this end, we want to apply Proposition~\ref{prop:cross_smooth} for sufficiently small~$r \leq \sigma q$.
We choose the functions
\beq \label{vdef}
v(\xi):=e^{-|\xi|}\qquad \text{and}
\qquad \tau(\xi):=\frac{1}{2}\,\Big( \frac{1}{m} + \frac{1}{|\xi| + \frac{m}{\alpha} } \Big)^{-1} \:.
\eeq
Let us verify that these functions satisfy the required conditions~\eqref{scales:eq}--\eqref{w:eq}
for large~$\alpha$.
First, a direct computation yields for any~$\xi \neq 0$
\[ \big|\partial_\xi \tau (\xi) \big| 
= \frac{\alpha^2 m^2}{2\, \big( m+ \alpha\,(|\xi|+m) \big)^2} \leq \frac{1}{2} \:,
\]
implying~\eqref{Lip:eq} with~$\nu=\frac{1}{2}$, if~$\bxi$ and~$\boldeta$ have the same sign. If they have different signs, we obtain
\[
| \tau (\bxi)- \tau(\boldeta) | \leq |\tau(\bxi)-\tau(0)| + |\tau(0)-\tau(\boldeta)| \leq \frac{1}{2}\, \big(|\bxi|+|\boldeta|\big) = \frac{1}{2}\, |\bxi-\boldeta|\:,
\]
proving~\eqref{Lip:eq}. Next, the estimate~\eqref{w:eq} follows from the estimates
\begin{gather*}
\tau(\xi) \leq \frac{m}{2} \\
\exp\Big( -\frac{m}{2} \Big) \leq e^{-\tau(\xi)} \leq
\frac{v(\xi)}{v(\eta)} = e^{-|\xi| + |\eta|} \leq e^{\tau(\xi)} \leq \exp\Big( \frac{m}{2} \Big) \:.
\end{gather*}
Finally, in order to prove~\eqref{scales:eq}, we compute the derivatives of the symbols~$\CA^{>}_\bullet$
as given by~\eqref{CAG} as well as~\eqref{Aadef} and~\eqref{Ainf} with the product and chain rules.
Since the function~$(1-\Theta_\alpha)$ and all its derivatives vanish unless~$|\xi| \geq
\sqrt{\log \alpha}/\alpha$, we know that, for large~$\alpha$, the factor~$|\xi|$ is bounded from below by~$|\xi| \gg m/\alpha$.
Therefore, each $\xi$-derivative of the symbol~$\CA_\bullet$ gives a scaling factor~$1/|\xi|$.
Using that, for large~$\alpha$,
\beq \label{tauscale}
\frac{1}{\tau(\xi)} = 2\,\Big( \frac{1}{m} + \frac{1}{|\xi| + \frac{m}{\alpha} } \Big) 
\geq \frac{2}{|\xi| + \frac{m}{\alpha}} \geq \frac{4}{|\xi|} \:,
\eeq
this can be bounded by a scaling factor~$1/\tau(\xi)$.
Each $\xi$-derivative acting on the cutoff function~$\Theta_\alpha$, on the other hand, 
gives a scaling factor~$\alpha/\sqrt{\log \alpha}$. In this case, we can use the fact that~$|\xi|$ 
is bounded from above by
\[ |\xi| \leq \frac{2 \alpha}{\sqrt{\log \alpha}} \:. \]
Combining this inequality with~\eqref{tauscale}, we conclude
that every $\xi$-derivative acting on the cutoff function gives again
the desired scaling factor~$1/\tau(\xi)$. We conclude that also the condition~\eqref{scales:eq} holds.

Applying Proposition~\ref{prop:cross_smooth}, we obtain
\[ \big\| \chi_\Lambda\, \Op_1(a_\alpha),(1-\chi_\Lambda) \big\|_{r}^r
\lesssim V_{r,1}(v,\tau)
=\int_{-\infty}^\infty \frac{e^{-r |\xi|}}{\tau(\xi)} \: d\xi = 2 \int_0^\infty \frac{e^{-r \xi}}{\tau(\xi)} \: d\xi \:. \]
In order to estimate the last integral, we split it into a sum of two integrals,
\beq \label{intsum}
\int_0^\infty \frac{e^{-r \xi}}{\tau(\xi)} \: d\xi = 
\int_0^{\frac{m}{\alpha}} \frac{e^{-r \xi}}{\tau(\xi)} \: d\xi + \int_{\frac{m}{\alpha}}^\infty \frac{e^{-r \xi}}{\tau(\xi)} \: d\xi
\eeq
In the first integral on the right, we can use the estimate
\[ \frac{1}{\tau(\xi)} = 2\,\Big( \frac{1}{m} + \frac{1}{\xi + \frac{m}{\alpha} } \Big) \leq 
\frac{2}{m} + \frac{\alpha}{m} \]
to conclude that the integral is uniformly bounded in~$\alpha$. In the last integral in~\eqref{intsum},
on the other hand, we employ the estimate
\[ \frac{1}{\tau(\xi)} = 2\,\Big( \frac{1}{m} + \frac{1}{\xi + \frac{m}{\alpha} } \Big) \leq 
\frac{2}{m} + \frac{2}{\xi} \]
to obtain
\[ \int_{\frac{m}{\alpha}}^\infty \frac{e^{-r \xi}}{\tau(\xi)} \: d\xi
\leq \frac{2}{m r} + 2
\int_{\frac{m}{\alpha}}^\infty \: e^{-r \xi}\: \frac{d\xi}{\xi} \lesssim \log \alpha \:. \]
Using the resulting bound for the function~$g(\alpha)$ in~\eqref{cond:chi_Lambda Opa (1-chi_Lambda)},
we can apply Lemma~\ref{lem:ErrorTermsAbstr} to obtain~\eqref{highfreq}.
\end{proof}

Using Proposition~\ref{prop:highfreq}, we can estimate the left side of~\eqref{eq:offdiag}
as follows,
\begin{align}
\lim_{\alpha\to\infty} & \frac{1}{\log\alpha} \:\Big| \tr\big(   D_\alpha(\eta_\varkappa, \Lambda, \CA_\alpha ) -   D_\alpha(\eta_\varkappa, \Lambda, \CA_\infty) \big) \Big| \notag \\
 &=\lim_{\alpha\to\infty}\frac{1}{\log\alpha}\: 
 \bigg( \Big| \tr \big( D_\alpha(\eta_\varkappa, \Lambda, \CA_\alpha ) -   D_\alpha(\eta_\varkappa, \Lambda, \CA_\alpha^>) \big) \Big| \notag \\
 &\qquad\qquad\;\;\quad 
 + \Big| \tr \big( D_\alpha(\eta_\varkappa, \Lambda, \CA_\infty ) -   D_\alpha(\eta_\varkappa, \Lambda, \CA_\infty^>) \big) \Big| \bigg) \notag \\
 & \leq \lim_{\alpha\to\infty}\frac{1}{\log\alpha}  \Big(  \big\|D_\alpha(\eta_\varkappa, \Lambda, \CA_\alpha ) -   D_\alpha(\eta_\varkappa, \Lambda, \CA^>_\alpha)\big\|_1 \notag \\
 &\qquad\qquad\;\quad 
 + \big\|D_\alpha(\eta_\varkappa, \Lambda, \CA_\infty ) -   D_\alpha(\eta_\varkappa, \Lambda, \CA^>_\infty)
 \big\|_1 \Big) \notag \\
& \leq \lim_{\alpha\to\infty}\frac{1}{\log\alpha}  \Big( \big\|\eta_\varkappa  \big(\chi_{\Lambda }\: \Opa(\CA_{\alpha})\:  \chi_{\Lambda } \big) - \eta_\varkappa \big(\chi_{\Lambda }\: \Opa(  \CA^>_\alpha)  \: \chi_{\Lambda } \big)\big\|_1
\label{t21} \\
& \qquad\qquad\quad\; 
+\big\|\eta_\varkappa  \big(\chi_{\Lambda }\: \Opa(\CA_\infty)\:  \chi_{\Lambda } \big) - \eta_\varkappa \big(\chi_{\Lambda }\: \Opa(  \CA^>_\infty)  \: \chi_{\Lambda } \big)\big\|_1 \label{t22} \\
& \qquad\qquad\quad\; +\big\|\chi_{\Lambda }\: \eta_\varkappa  \big(\Opa(\CA_{\alpha})\big)   \chi_{\Lambda } - \chi_{\Lambda }\: \eta_\varkappa \big(\Opa(  \CA^>_\alpha)\big)  \: \chi_{\Lambda } \big\|_1 \Big) \label{t23} \\
& \qquad\qquad\quad\; +\big\|\chi_{\Lambda }\: \eta_\varkappa  \big(\Opa(\CA_{\infty})\big)   \chi_{\Lambda } - \chi_{\Lambda }\: \eta_\varkappa \big(\Opa(  \CA^>_\infty)\big)  \: \chi_{\Lambda } \big\|_1 \Big) \:. \label{t24}
\end{align}

Our task is to estimate all the obtained summands. We begin with the summands~\eqref{t23}
and~\eqref{t24}. According to~\eqref{PiOpa}, the operator~$\Opa(  \CA^>_\alpha)$ is a multiplication
operator in momentum space.
Therefore, its spectral calculus reduces to that of $2 \times 2$-matrices, which can be worked out explicitly.
Using that the square bracket in~\eqref{Aadef} is a matrix with eigenvalues zero and two,
a straightforward computation gives
\[ \eta_\varkappa \big(\Opa(  \CA_\alpha)\big) = \Opa(  \CB_\alpha ) \]
with
\[ \CB_\alpha(\xi) := \frac{1}{2}\: \eta_\varkappa \Big( e^{-\sqrt{\xi^2+(m/\alpha)^2}}\, \Big)
\bigg(\1_{\C^2} - \frac{1}{  \sqrt{\xi^2 + (m/\alpha)^2}} \begin{pmatrix} -\xi & m/\alpha \\ m/\alpha & \xi \end{pmatrix} \bigg) \:. \]
The operator~$\Opa(  \CA^>_\alpha)$ can be treated similarly. Indeed, using~\eqref{CAG}, we obtain
\[ \eta_\varkappa \big(\Opa(  \CA^>_\alpha)\big) = \Opa(  \CB_\alpha^> ) \]
with
\[ \CB_\alpha(\xi) := \frac{1}{2}\: \eta_\varkappa \Big( (1-\Theta_\alpha)(\xi)\: e^{-\sqrt{\xi^2+(m/\alpha)^2}}\, \Big)
\bigg(\1_{\C^2} - \frac{1}{  \sqrt{\xi^2 + (m/\alpha)^2}} \begin{pmatrix} -\xi & m/\alpha \\ m/\alpha & \xi \end{pmatrix} \bigg) \:. \]
Similarly, using~\eqref{Ainf} and again~\eqref{CAG}, we obtain
\[ \eta_\varkappa \big(\Opa(  \CA_\infty)\big) = \Opa(  \CB_\infty ) \qquad \text{and} \qquad
\eta_\varkappa \big(\Opa(  \CA_\infty^>)\big) = \Opa(  \CB_\infty^> ) \]
with
\begin{align*}
\CB_\infty(\xi) &:= \frac{1}{2}\: \eta_\varkappa \big( e^{-|\xi|}\, \big)\;
\begin{pmatrix}  \chi_{\mathbb{R}^+}(\xi) & 0\\ 0 &  \chi_{\mathbb{R}^-}(\xi) \end{pmatrix} \\
\CB_\infty^>(\xi) &:= \frac{1}{2}\: \eta_\varkappa \big( (1-\Theta_\alpha)(\xi)\: e^{-|\xi|}\, \big)\;
\begin{pmatrix}  \chi_{\mathbb{R}^+}(\xi) & 0\\ 0 &  \chi_{\mathbb{R}^-}(\xi) \end{pmatrix} \:.
\end{align*}
The detailed form of these symbols will not be needed in what follows.
The important point is that the difference~$\CB_\bullet - \CB_\bullet^>$ is bounded and vanishes
for large frequencies; more precisely,
\beq \label{Bprop}
\sup_\xi \big\|(\CB_\bullet-\CB_\bullet^>)(\xi) \big\| \leq 2 \qquad \text{and} \qquad
(\CB_\bullet-\CB_\bullet^>)(\xi) = 0 \quad \text{if~$\displaystyle |\xi| \geq \frac{2 \sqrt{\log \alpha}}{\alpha}$} \:.
\eeq
(where the bullet stands again for~$\alpha$ or~$\infty$).
Using this notation, the summands~\eqref{t23} and~\eqref{t24} can be estimated by
\beq \label{term2}
\Big\|\chi_{\Lambda }\Big( \eta_\varkappa  \big(\Opa(\CA_\bullet \big)
- \eta_\varkappa \big(\Opa(  \CA^>_\bullet)\big) \Big) \: \chi_{\Lambda } \Big\|_1
\leq \Big\| \chi_\Lambda\: \Opa(\CB_\bullet - \CB_\bullet^>) \: \chi_\Lambda \big\|_1 \:.
\eeq
This norm is estimated in the following proposition.
\begin{Prp} \label{prpfinal}
For any~$\sigma \in (0,2)$,
\begin{align*}
 \lim_{\alpha\to\infty}\frac{1}{\log\alpha} \: \big\|\chi_{\Lambda}\, \Opa(\CA^<_\bullet)\, \chi_{\Lambda }
 \big \|_{\sigma}^{\sigma} &= 0 \\
 \lim_{\alpha\to\infty}\frac{1}{\log\alpha} \:\big\|\chi_{\Lambda}\, \Opa(  \CB_\bullet - \CB_\bullet^>
 )\,\chi_{\Lambda}  \big\|_1 &=0\,.
\end{align*}
\end{Prp}
\begin{proof} We want to apply Proposition~\ref{prp72}.
Since this estimate involves~$\text{Op}_1$ instead of~$\Opa$, we need to rescale in momentum space,
Indeed,
\[ \Big\|  \chi_{\Lambda}\, \Opa(\CA^<_{\bullet})\, \chi_{\Lambda } \Big \|_{\sigma}
= \big\| \chi_{\Lambda}\, \text{Op}_1\big(\CA^{<}_\bullet(\alpha^{-1} \:\cdot\,) \big)\big\|_\sigma \:. \]
Equivalently, this means that we need to express the symbols in terms of our original momentum
variable~$k$ (see~\eqref{xiscale}). Applying Proposition~\ref{prp72},
for any~$\sigma\in (0,2)$ we can estimate the Schatten norm by
\[ \Big\|  \chi_{\Lambda}\, \Opa(\CA^<_{\bullet})\, \chi_{\Lambda } \Big \|_{\sigma}
\leq C \:\big|\chi_\Lambda \big|_\sigma \: \big| \CA^{<}_\bullet(\alpha^{-1} \:\cdot\,) \big|_\sigma \:, \]
where~$| \cdot |_\sigma$ is again the norm~\eqref{sigmanorm}.
Using that the matrix-valued symbol is uniformly bounded, $\|\CA^{<}_\bullet(\xi)\| \leq c$,
we thus obtain the estimate
\[ \Big\|  \chi_{\Lambda}\, \Opa(\CA^<_{\bullet})\, \chi_{\Lambda } \Big \|_{\sigma}^\sigma
\leq C^\sigma \:\big|\chi_\Lambda \big|_\sigma^\sigma \int_1^{2 \sqrt{\log \alpha}} c^\sigma\: dk
\lesssim \sqrt{ \log \alpha } \:. \]
Dividing by~$\log \alpha$, the resulting expression tends to zero as~$\alpha \rightarrow \infty$.

The symbol~$\CB_\bullet - \CB^>_\bullet$ can be treated in the same way, because
\begin{align*}
\big\| &\chi_{\Lambda}\, \Opa(  \CB_\bullet - \CB^>_\bullet )\,\chi_{\Lambda}  \big\|_1
\leq C \:\big|\chi_\Lambda \big|_1 \: \big| (\CB_\bullet - \CB_\bullet^>)(\alpha^{-1} \:\cdot\,) \big|_1 \\
&\overset{\eqref{Bprop}}{\leq} C \:\big|\chi_\Lambda \big|_1 \int_1^{2 \sqrt{\log \alpha}} 2\: dk
\leq 2 C \:\big|\chi_\Lambda \big|_1\: \sqrt{\log \alpha} \:. 
\end{align*}
This concludes the proof.
\end{proof}

It remains to estimate the first two summands~\eqref{t21} and~\eqref{t22}.
We point out that, here, the spectral calculus cannot be performed explicitly, because
operators like~$\chi_{\Lambda }\: \Opa(\CA_\infty)\:  \chi_{\Lambda}$ are not multiplication
operators in momentum space, but rather products of multiplication and convolution operators.
For this reason, we need to use the estimate of Proposition~\ref{prop:szego}.
Choosing~$\sigma\in (2/3, 1)$, we obtain
\begin{align}
&\big\| \eta_\varkappa   \big(\chi_{\Lambda }\: \Opa(\CA_{\bullet})  \chi_{\Lambda } \big) - \eta_\varkappa \big(\chi_{\Lambda }\: \Opa(  \CA^>_\bullet)  \: \chi_{\Lambda } \big) \big \|_1 \notag \\
& \leq C\, \big\|  \chi_{\Lambda}\, \big( \Opa(\CA_{\bullet}- \CA^>_\bullet) \big)\, \chi_{\Lambda } \big \|_{\sigma}^\sigma
= C\, \big\|  \chi_{\Lambda}\, \Opa(\CA^<_{\bullet})\, \chi_{\Lambda } \big \|_{\sigma}^{\sigma} \:.
\label{term1}
\end{align} 
The last norm has been estimated in Proposition~\ref{prpfinal}. This concludes the proof of Theorem~\ref{thmmain}.
 
\Thanks{{{\em{Acknowledgments:}} We would like to thank Alexander Sobolev for helpful discussions.
We are grateful to the referee for helpful comments and suggestions.
M.L.\ gratefully acknowledges support by the Studienstiftung des deutschen Volkes and the Marianne-Plehn-Programm. We would like to thank the ``Universit\"atsstiftung Hans Vielberth'' for support.
S.M.\ is partially supported by INFN and by MIUR Excellence Department Project 2023-2027 awarded to the Department of Mathematics of the University of Genoa. This work was written within the activities of the GNFM group of INdAM.

\providecommand{\bysame}{\leavevmode\hbox to3em{\hrulefill}\thinspace}
\providecommand{\MR}{\relax\ifhmode\unskip\space\fi MR }
\providecommand{\MRhref}[2]{%
  \href{http://www.ams.org/mathscinet-getitem?mr=#1}{#2}
}
\providecommand{\href}[2]{#2}

\end{document}